\documentclass[letterpaper, 10pt, conference]{ieeeconf}

\IEEEoverridecommandlockouts                              	

\usepackage{amssymb,amsmath,amsfonts,amscd}
\usepackage{algorithm,algpseudocode} 
\usepackage{cite}
\usepackage{float}
\usepackage{epsfig}
\usepackage{graphicx}
\usepackage{subfigure}
\usepackage{url}
\usepackage{verbatim}
\usepackage{xspace}
\usepackage[usenames,dvipsnames]{color}
\usepackage{soul}
\usepackage{flushend}
\usepackage{mathtools}
\usepackage[table]{xcolor}
\usepackage{hyperref}
\usepackage{microtype}      
\usepackage{dsfont}
\usepackage{caption}
\usepackage{array}
\makeatletter
\newcommand{\thickhline}{%
    \noalign {\ifnum 0=`}\fi \hrule height 1pt
    \futurelet \reserved@a \@xhline
}
\newcolumntype{"}{@{\hskip\tabcolsep\vrule width 1pt\hskip\tabcolsep}}
\makeatother

\usepackage{cuted} 

\usepackage{pgfplots}
\pgfplotsset{compat=1.17}  

\usepackage{array}
\newcommand{\PreserveBackslash}[1]{\let\temp=\\#1\let\\=\temp}
\newcolumntype{C}[1]{>{\PreserveBackslash\centering}p{#1}}

\floatstyle{ruled}
\newfloat{model}{H}{mod}
\floatname{model}{\footnotesize Model}
\newfloat{notatio}{H}{not}
\floatname{notatio}{\footnotesize Notation}

\newenvironment{varalgorithm}[1]
  {\algorithm\renewcommand{\thealgorithm}{#1}}
  {\endalgorithm}

\newenvironment{list3}{
	\begin{list}{$\bullet$}{%
			\setlength{\itemsep}{0.05cm}
			\setlength{\labelsep}{0.2cm}
			\setlength{\labelwidth}{0.3cm}
			\setlength{\parsep}{0in} 
			\setlength{\parskip}{0in}
			\setlength{\topsep}{0in} 
			\setlength{\partopsep}{0in}
			\setlength{\leftmargin}{0.22in}}}
	{\end{list}}

\newenvironment{list4}{
	\begin{list}{$\bullet$}{%
			\setlength{\itemsep}{0.05cm}
			\setlength{\labelsep}{0.2cm}
			\setlength{\labelwidth}{0.3cm}
			\setlength{\parsep}{0in} 
			\setlength{\parskip}{0in}
			\setlength{\topsep}{0in} 
			\setlength{\partopsep}{0in}
			\setlength{\leftmargin}{0.16in}}}
	{\end{list}}

\usepackage{url,changebar,bm,xspace,dsfont}
\let\mathbb=\mathds 
\def\day{\mathrm{day}} 

\newtheorem{prop}{Proposition}
\newtheorem{assum}{Assumption}
\newtheorem{example}{\bfseries Example}

\newtheorem{remark}{Remark}

\usepackage{capt-of}

\hypersetup{
    colorlinks = true,
    linkcolor = [rgb]{0,0,1},
    anchorcolor = [rgb]{0,0,1},
    citecolor = [rgb]{0.9,0.5,0},
    filecolor = [rgb]{0,0,1},
    urlcolor = [rgb]{0.9,0.5,0},
    bookmarksopen = true,
    bookmarksnumbered = true,
    breaklinks = true,
    linktocpage,
    colorlinks = true,
    linkcolor = [rgb]{0.2,0.6,0.2},
    urlcolor  = [rgb]{0,0,1},
    citecolor = [rgb]{0.9,0.5,0},
    anchorcolor = [rgb]{0.2,0.6,0.2},
}

\begin{document}

\title{\LARGE \bf HARQ-based Quantized Average Consensus \\ over Unreliable Directed Network Topologies}

\author{Neofytos Charalampous, Evagoras Makridis, Apostolos~I.~Rikos, and Themistoklis~Charalambous
\thanks{N. Charalambous, E. Makridis, and T. Charalambous are with the Department of Electrical and Computer Engineering, School of Engineering, University of Cyprus, 1678 Nicosia, Cyprus.  E-mails: {\tt~surname.name@ucy.ac.cy}. T. Charalambous is also a Visiting Professor at the Department of Electrical Engineering and Automation, School of Electrical Engineering, Aalto University, 02150 Espoo, Finland.}
\thanks{Apostolos~I.~Rikos is with the Artificial Intelligence Thrust of the Information Hub, The Hong Kong University of Science and Technology (Guangzhou), Guangzhou, China. He is also affiliated with the Department of Computer Science and Engineering, The Hong Kong University of Science and Technology, Clear Water Bay, Hong Kong. E-mail: {\tt~apostolosr@hkust-gz.edu.cn}.}
\thanks{The work of E. Makridis and T. Charalambous has been partly funded by MINERVA, a European Research Council (ERC) project funded under the European Union's Horizon 2022 research and innovation programme (Grant agreement No. 101044629).}
}

\maketitle

%
%
%
%
\begin{abstract}
In this paper, we propose a distributed algorithm (herein called HARQ-QAC) that enables nodes to calculate the average of their initial states by exchanging quantized messages over a directed communication network. In our setting, we assume that our communication network consists of unreliable communication links (\emph{i.e.,} links suffering from packet drops). 
For countering link unreliability our algorithm leverages narrowband error-free feedback channels for acknowledging whether a packet transmission between nodes was successful. 
Additionally, we show that the feedback channels play a crucial role in enabling our algorithm to exhibit finite-time convergence. We analyze our algorithm and demonstrate its operation via an example, where we illustrate its operational advantages. Finally, simulations corroborate that our proposed algorithm converges to the average of the initial quantized values in a finite number of steps, despite the packet losses. This is the first quantized consensus algorithm in the literature that can handle packet losses and converge to the average. Additionally, the use of the retransmission mechanism allows for accelerating the convergence. 
\end{abstract}


%
%
%
%
\section{Introduction}\label{sec:intro}

In today’s interconnected world, multi-agent systems are vital for managing networks efficiently. 
For this reason, there has been a surge in interest regarding the control and coordination of networks comprising multiple agents such as groups of sensors, 
smart grids, 
social networks, 
and mobile autonomous agents. 
One problem of particular interest over multi-agent systems is reaching \textit{consensus} in a distributed fashion. 
Consensus is the process of reaching an agreement among distributed agents to converge towards a common decision. 
Its applications span a wide area including distributed optimization \cite{NIPS2012_d240e3d3,7447011,2020:Nedich, 2023:Rikos_Johan_IFAC}, 
distributed reinforcement learning~\cite{STANKOVIC2023110922}, cooperative control~\cite{8458142}, and networked control systems~\cite{7949013}. 
Additionally, enabling distributed computations among agents is essential for scalability, fault tolerance, and adaptability to dynamic environments \cite{nedic2018distributed}. 

One special case of consensus is distributed averaging over directed network topologies with bandwidth constraints and (probably) unreliable communication links. 
In distributed averaging nodes are initially endowed with a numerical state, and aim to calculate the average of their states. 
Early efforts primarily rely on the assumption that communication channels of networks have unlimited capacity and nodes can exchange exchange real-valued messages; this problem has been extensively studied across various research directions; see, e.g.,  \cite{qin2016recent,2018:BOOK_Hadj} and references therein. Some examples include the presence of unreliable communication links in the network \cite{hadjicostis2015robust, makridis2023harnessing}, permanent variations in the network's composition (\emph{i.e.,} arrival or departure of nodes) \cite{2023:Galland_Hendrickx}, and the dynamic nature of the network topology \cite{RIKOS2022110621}. 
Attention has been also given to the case where the network consists of unreliable communication links which significantly enhances robustness \cite{4434917, 2009:Zampieri, hadjicostis2015robust, makridis2023harnessing}. 

However, the assumption of infinite capacity is not true in practice, because digital channels are subject to bandwidth constraints and, hence, only a finite number of bits 
can be transmitted along a channel. 
%
Quantization is vital for resource-constrained wireless networks, e.g., a network consisting of Internet of Things (IoT) devices.
For improving operational efficiency, recently there is an increasing interest on works where nodes calculate the average of their states by exchanging quantized messages \cite{2011:Cai_Ishii, 2012:Lavaei_Murray, 2013:GarciaCasbeer, 2016:Chamie_Basar,  2020:Rikos_Quant_Cons, RIKOS2022110621, rikos2024finite}.
However, none of the aforementioned works considered packet drops. The only works that considered quantized average consensus for directed graphs and included switching topologies are 
\cite{SIAM:2013} and \cite{LI:2014}. However, in \cite{SIAM:2013} the network topology is assumed to be balanced, whereas in \cite{LI:2014}, although the network is directed, consensus is reached, but it does not necessarily correspond to the average. To the best of the authors' knowledge, the problem of quantized average consensus over directed network topologies in the presence of unreliable communication links, such as packet-dropping links, has not yet been satisfactorily addressed. 
This problem is of paramount significance as it paves the way for developing robust algorithms capable of operating effectively in realistic wireless communication environments. 

\textbf{Main Contributions.} 
Motivated by the aforementioned gap in the literature, in this paper, we present the first distributed algorithm that enables nodes to achieve quantized average consensus over directed network topologies in the presence of unreliable communication links. This is achieved by invoking a Hybrid Automatic Repeat reQuest (HARQ) error control protocol, which  allows nodes to exploit incoming error-free acknowledgement feedback signals to know whether a packet has arrived, but also to achieve fast discrete-time asymptotic average consensus in the presence of packet retransmissions (information delays), and packet-dropping  links (information loss).
Our main contributions are the following: 
\begin{list4}
    \item We present a distributed algorithm that achieves quantized average consensus in finite time in the presence of unreliable packet links in the network; see Algorithm~\ref{alg:arq_quantized}. 
    \item Leveraging HARQ error-free feedback channels, our algorithm is able to maintain running sums that do not explode as time increases (similarly to~\cite{makridis2023harnessing} and unlike~\cite{hadjicostis2015robust}) and mitigates the impact of packet losses, since retransmissions decrease significantly the packet error probability; see equation~\eqref{eq:harq_packet_error_probability}.
    \item We show that our algorithm is able to converge to the average of the initial states in finite time despite the presence of network unreliabilities; see Proposition~\ref{proof_theorem}. 
    \item We validate the correctness of our algorithm via extended simulations over various scenarios; see Section~\ref{sec:results}.
\end{list4}


%
%
%
%
\section{Notation and Preliminaries}\label{sec:preliminaries}

\textbf{Mathematical Notation.} 
The sets of real, rational, integer, and natural numbers are denoted as $\mathbb{R}, \mathbb{Q}, \mathbb{Z}$, and $\mathbb{N}$, respectively. 
The set of nonnegative integer numbers and the set of positive real numbers are denoted as $\mathbb{Z}_+, \mathbb{R}_{\geq 0}$. 
The nonnegative orthant of the $n$-dimensional real space $\mathbb{R}^n$ is denoted as $\mathbb{R}_{\geq 0}$. 
Matrices are denoted by capital letters (e.g.,  $A$), and vectors by small letters (e.g., $x$). 
The transpose of a matrix $A$ and a vector $x$ are denoted as $A^\top$, $x^\top$, respectively.
For any real number $a \in \mathbb{R}$, the greatest integer less than or equal to $a$ is denoted as $\lfloor a \rfloor$, and the least integer greater than or equal to $a$ is denoted as $\lceil a \rceil$. 
For a matrix $A \in \mathbb{R}^{n \times n}$, the entry in row $i$ and column $j$ of matrix  $A$ is denoted as $A_{ij} \in \mathbb{R}$.
The all-ones vector and the identity matrix of appropriate dimensions are denoted as $\mathbb{1}$ and $\mathbb{I}$, respectively. 
The Euclidean norm of a vector $x$ is represented by $\| x \|$. 

\textbf{Graph-Theoretic Notions.} 
The communication network is modeled as a directed graph (or digraph) denoted as $\mathcal{G}_d = (\mathcal{V}, \mathcal{E})$. 
It comprises $n$ nodes ($n \geq 2$), communicating solely with their immediate neighbors, and remains static over time. 
A directed edge from node $v_i$ to node $v_j$ is denoted by $\varepsilon_{ji} \in \mathcal{E}$, signifying that node $v_j$ can receive information from node $v_i$ at time step $k$ (but not vice versa). 
The node set is $\mathcal{V} = \{ v_1, v_2, ..., v_n \}$. 
The edge is $\mathcal{E} \subseteq \mathcal{V} \times \mathcal{V} \cup \{ \varepsilon_{jj} \ | \ v_i \in \mathcal{V} \}$ (each node has a virtual self-edge). 
The number of nodes is denoted as $| \mathcal{V} | = n$ and the number of edges as $| \mathcal{E} | = m$. 
The in-neighbors of $v_j$ are denoted by $\mathcal{N}_j^- = \{ v_i \in \mathcal{V} | \varepsilon_{ji} \in \mathcal{E}\}$, and they represent the nodes that can directly transmit information to $v_j$. 
The out-neighbors of $v_j$ are denoted by $\mathcal{N}_j^+ = \{ v_l \in \mathcal{V} | \varepsilon_{lj} \in \mathcal{E}\}$ and they represent the nodes that can directly receive information from $v_j$. 
The in-degree and out-degree of $v_j$ are denoted by $\mathcal{D}_j^- = | \mathcal{N}_j^- |$ and $\mathcal{D}_{j}^+ = | \mathcal{N}_{j}^+ |$, respectively. The diameter $D$ of a digraph is the longest shortest path between any two nodes $v_l, v_i \in \mathcal{V}$. 
A directed path from $v_i$ to $v_l$ of length $t$ exists if there's a sequence of nodes $i \equiv l_0, l_1, \dots, l_t \equiv l$ such that $\varepsilon_{l_{\tau+1} l_{\tau}} \in \mathcal{E}$ for $ \tau = 0, 1, \dots , t-1$. 
A digraph is \textit{strongly connected} if there's a directed path from every node $v_i$ to every node $v_l$ for all $v_i, v_l \in \mathcal{V}$. 

\textbf{Transmission Policy.} 
Each node possesses knowledge of its out-neighbors and can directly transmit messages to each of them. 
However, it may not receive messages directly from these out-neighbors due to the network's directed nature.
In our proposed distributed algorithm, each node $v_j$ assigns a \textit{unique order} in the set ${1,..., \mathcal{D}_j^+}$ to each of its outgoing edges $\varepsilon_{lj}$, where $v_l \in \mathcal{N}^+_j$. 
The order of link $ \varepsilon_{lj}$ for node $v_j$ is denoted by $P_{lj}$ (such that $\{P_{lj} \; | \; v_l \in \mathcal{N}^+_j\} = \{1,..., \mathcal{D}_j^+\}$). 
During the execution of the proposed distributed algorithm, this unique order facilitates each node $v_j$ in transmitting messages to its out-neighbors in a \textit{round-robin} manner adhering to the predefined order. 

\textbf{Modeling Packet Drops and Acknowledgements.} 
Each communication link in our distributed network is unreliable and is subject to packet drops. 
Packet drops can occur due to various reasons such as network congestion, channel impairments, or other factors. For a link $\varepsilon_{ji}$, the probability of a packet being dropped (or lost) at time step $k$ is denoted by $p_{ji}[k]$, where $0 \leq p_{ji}[k] < 1$. 
In our analysis, we assume that packet drops occur independently and are identically distributed across different links and time instances. We model the occurrence of packet drops as a Bernoulli random variable. Specifically, with $x_{ji}[k] = 1$ we denote a successful transmission from node $v_i$ to node $v_j$ at time step $k$, while with $x_{ji}[k] = 0$, we represent a failed transmission. The probability distribution of $x_{ji}[k]$ is given by
\begin{equation}\label{dropsmodel}
Pr\{ x_{ji}[k] = m \} = \begin{cases}
    p_{ji}[k], & \text{if~} m = 0,\\
    1 - p_{ji}[k], & \text{if~} m = 1.
\end{cases}
\end{equation}
The data receiver $v_j$ determines $x_{ji}[k]$ through an Automatic Repeat reQuest (ARQ) error control protocol. ARQ is used in various message transmission protocols (e.g., Transmission Control Protocol (TCP), High-Level Data Link, and others). 
ARQ relies on retransmitting previously failed transmissions with the aid of acknowledgement messages transmitted via a narrowband feedback channel from the data receiver to the data transmitter. Thus, data transmitter decides whether to retransmit data packets based on the acknowledgement message it received from the data receiver, which is given by
\begin{align}
    f_{ji}[k]=
    \begin{cases}
        0,& \text{if~} x_{ji}[k]=0,\\ 
        1,& \text{otherwise}.
    \end{cases}
\end{align}
The acknowledgement message is transmitted over a narrowband feedback channel, for which we make the following assumption.
\begin{assum}\label{assum_feedback_channel}
Each directed link $\varepsilon_{ji} \in \mathcal{E}$ from node $v_i$ to node $v_j$ is accompanied by a narrowband error-free feedback channel from node $v_j$ to node $v_i$.
\end{assum}  
These feedback channels cannot support data transmission as they usually comprise messages of $1$ bit only. For this reason, they are (assumed to be) error-free; see, e.g., \cite{JSAC:2008-errorfreechannels,2022:Makridis_Themis_Hadj}. 

\textbf{Decoding Error Reduction via Retransmissions.} 
In our paper, we adopt a communication protocol in which the decoding error reduces when a packet is retransmitted, such that the probability of packet drop reduces exponentially with the number of retransmissions \cite{HARQ:2010}. 
More specifically, consider that a fresh packet is initially transmitted at time $k$ over the link $\varepsilon_{ji}\in\mathcal{E}$. Then, the probability of error when retransmitting the same packet for $r\in\{1,\ldots,\bar{\tau}_{ji}\}$ times using the HARQ mechanism, can be approximated with the following packet error probability model \cite{HARQ:2019-Gunduz}: 
\begin{align}\label{eq:harq_packet_error_probability}
    p_{ji}[k+r]=p_{ji}[k] \lambda^{r},
\end{align}
where $\lambda \in (0,1]$ is a parameter that depends on the HARQ protocol used and the channel model, and $\bar{\tau}_{ji}$ is the number of allowable retransmissions. Clearly, the smaller the value of $\lambda$, the faster the error probability decreases with retransmissions. Thus, with a higher number of allowable retransmissions, the packet error probability becomes smaller. Setting $\lambda=1$, the HARQ protocol reduces to the ARQ.



%
%
%
%
\section{Problem Formulation}\label{sec:probForm}

Let us consider a network of nodes (e.g., wireless sensors, robots, drones), communicating over a wireless channel modeled as a digraph $\mathcal{G}_d = (\mathcal{V}, \mathcal{E})$. 
Each node $v_j \in \mathcal{V}$ is assumed to have an initial state $y_j[0] \in \mathbb{Z}$\footnote{This assumption states that the measurements have undergone quantization already. 
This is a standard assumption in quantized consensus~\cite{2007:Basar, 2011:Cai_Ishii,2018:RikosHadj_CDC}.}. 
The communication links among nodes are not reliable due to the wireless nature of the channels (e.g., due to fading, and interference), and as a consequence packet drops may occur. 
Specifically, each communication link $\varepsilon_{ji}$ is unreliable and subject to packet drops (or packet losses) with probability $p_{ji}[k]$ at time step $k$. 
Additionally, in our network each directed link $\varepsilon_{ji}$ from node $v_i$ to node $v_j$ is accompanied by a narrowband error-free feedback channel from node $v_j$ to node $v_i$ (see Assumption~\ref{assum_feedback_channel}). 
In our paper, we aim at designing a distributed algorithm that enables nodes to calculate the average of their initial states, $\hat{y}\triangleq \frac{\sum_{v_i\in\mathcal{V}}y_i[0]}{n}$, in a finite number of steps, with quantized information exchange over packet-dropping directed links.

\section{Main Results}\label{sec:main_results}

In this section, we describe the operation of our proposed HARQ-based Quantized Average Consensus Algorithm, (herein abbreviated as HARQ-QAC algorithm). 
The proposed algorithm, unlike other works in the literature, manages potential packet reception failures by leveraging feedback from channels and integrating data from earlier packet transmissions. 
Our proposed algorithm is summarized in Algorithm~\ref{alg:arq_quantized} and described in the following four steps: 
\begin{list3}
\item[\textbf{1.}] Each node $v_j$ maintains an initial state $y_j[0]=V_j$ (it is determined by the information on which the network of nodes must reach average consensus). 
The initial value of the auxiliary variable $z_j[k]$ is set to $z_j[0] = 1$. 
Additionally, to track information that may be lost due to packet drops, each node maintains the accumulated $y_j$ and $z_j$ masses that node $v_j$ wants to transmit to each out-neighbor $v_l \in \mathcal{N}_j^{+}$. 
These masses are denoted by {$y_{lj}^b[k]$} and {$z_{lj}^b[k]$}, and initialized at $y_{lj}^b[0] = 0=z_{lj}^b[0] = 0$.
Similarly, each node maintains the state variables $y_j^s[k]$ and $z_j^s[k]$  which are initialized at $y_j^s[k]=y_j[0]$ and $z_j^s[k]=1$, respectively, and enable node $v_j$ to calculate the average of the initial states.
\item[\textbf{2.}] Each node $v_j$ assigns to each of its outgoing edges a unique order $P_{lj}$ to each $v_l\in {\mathcal{N}}_j^{+}$ from $1$ to $D_j^{+}$. 
Initially, it selects $v_l\in {\mathcal{N}}_j^{+}$ for which $P_{lj}=1$ and transmits $y_j[0]$ and $z_j[0]$ to this out-neighbor, by following the predetermined order. 
In subsequent transmissions to an out-neighbor, it resumes from the last used outgoing edge, cycling through the edges in a round-robin manner.
\item[\textbf{3.}] When a packet sent by an in-neighbor $v_i \in \mathcal{N}_j^{-}$ arrives at the receiving node $v_j$ with errors, the corresponding buffers $y_{ji}^b[k+1]$ and {$z_{ji}^b[k+1]$} are updated, whenever the actual packet retransmissions over link {$\varepsilon_{ji}$}, \emph{i.e.,} $\tau_{ij}[k]$, have not reached the maximum retransmission limit $\tau_{ij}[k] < \bar{\tau}_{ij}$ imposed by the HARQ protocol. 
In contrast, if the packet is error-free, these buffers are reset to zero. 
In case the retransmission limit of a packet has been reached, the running sum mechanism is activated (see lines 10--12 of Algorithm~\ref{alg:arq_quantized}).
\item[\textbf{4.}] Node $v_j$ receives the mass variables $y_i[k]$ and $z_i[k]$ from its in-neighbors $v_i$ $\in \mathcal{N}_{j}^{-}$ and sums them along with its stored mass variables $y_j[k]$ and $z_j[k]$ using the equations in lines 17, 18, in Algorithm~\ref{alg:arq_quantized}. 
It then compares the total accumulated mass against its current state variables $y_j^s[k]$ and $z_j^s[k]$ according to the event-trigger conditions \textbf{C1} and \textbf{C2}. 
More specifically, in \textbf{C1} node $v_j$ checks if its stored $z_j[k]$ mass variable is greater than its stored $z_j^s[k]$ state variable. 
Also, in \textbf{C2} it checks if $z_j[k]$ is equal to $z_j^s[k]$, and if this holds then it checks if its stored $y_j[k]$ mass variable is greater than its stored $y_j^s[k]$ state variable. 
If neither condition \textbf{C1} nor condition \textbf{C2} hold, it performs no actions. 
Otherwise, If at least one of \textbf{C1} and \textbf{C2} hold, $v_j$ updates its state variables to be equal to the stored mass variables.
Next, it selects the next out-neighbor $v_l$ $\in \mathcal{N}_{j}^{-}$ based on the predetermined priority, to transmit its values as in lines 25, 26, in Algorithm~\ref{alg:arq_quantized} (and then reset as in line 27).
\end{list3}

\begin{varalgorithm}{1}
\caption{HARQ-based Quantized Average Consensus (HARQ-QAC) algorithm for node $v_j$} 
\label{alg:arq_quantized}
\begin{algorithmic}[1]
\State \textbf{Input:} $y_j[0] \in \mathbb{Z}$.  
\State \textbf{Initialization:} Set $y_j^s[0] = y_j[0]$, $z_j[0] = 1$, $z_j^s[0] = 1$, $\mathrm{tr}_j=1$, and {$z_{lj}^b[0] = 0$}, {$y_{lj}^b[0] = 0$, $\tau_{lj}[0] = 0$, $f_{lj}[0] = 0$}, for all $v_l\in\mathcal{N}^+_j$. Assign a unique order $P_{lj}$ to each $v_l\in {\mathcal{N}}_j^{+}$ from 1 to $D_j^{+}$; then, select  $v_l\in {\mathcal{N}}_j^{+}$ for which $P_{lj}=1$ and transmit $y_j[0]$ and $z_j[0]$ to out-neighbor $v_l$.
\For{$k = 0, 1, 2, \ldots$}
    \State \textbf{Receive feedback from} \( v_l \in \mathcal{N}_j^{{+}} \) based on $P_{lj}$
    \State \textbf{if}  $f_{lj}[k] = 0 \land \tau_{lj}[k] < \bar{\tau}_{lj}$
    \State \quad $y_{lj}^b[k+1] = y_j[k]$
    \State \quad $z_{lj}^b[k+1] = z_j[k]$
    \State \quad $y_j[k]=0$, $z_j[k]=0$, $\tau_{lj}[k+1] = \tau_{lj}[k] + 1$
    \State \textbf{else if} $f_{lj}[k] = 0 \land \tau_{lj}[k] = \bar{\tau}_{lj}$
    \State \quad $y_{j}[k+1] = y_{j}[k] + y_{lj}^b[k]$
    \State \quad $z_{j}[k+1] = z_{j}[k] + z_{lj}^b[k]$ 
    \State \quad $y_{lj}^b[k]  = 0, z_{lj}^b[k] = 0, \tau_{lj}[k] = 0$ 
    \State \textbf{else}
    \State \quad $y_{lj}^b[k] = 0, z_{lj}^b[k] = 0, \tau_{lj}[k] = 0$ 
    \State \textbf{end}
    \State \textbf{Receive information from} \( v_i \in \mathcal{N}_j^{{-}} \)
    \State \quad $y_{j}[k+1] = y_{j}[k+1] + \sum_{v_i\in\mathcal{N}_j^-} f_{ji}[k] y_i[k]$ 
    \State \quad $z_{j}[k+1] = z_{j}[k+1] + \sum_{v_i\in\mathcal{N}_j^-} f_{ji}[k] z_i[k]$ 
    \State \textbf{Send feedback to} \( v_i \in \mathcal{N}_j^{-} \): $f_{ji}[k]$
    \State \textbf{C1:} $z_j[k+1] > z_j^s[k]$
    \State \textbf{C2:} $z_j[k+1] = z_j^s[k]$ $\land$ $y_j[k+1]$ $\geq$ $y_j^s[k]$
    \State \textbf{if} \textbf{C1} $\lor$ \textbf{C2}
    \State \quad $y_j^s[k+1] = y_j[k+1], z_j^s[k+1] = z_j[k+1]$
    \State \quad  $\mathrm{tr}_j$ = $(\mathrm{tr}_j\mod D^{+}_j) + 1 $
    \State \quad\textbf{Transmit to} $v_l \in \mathcal{N}_j^{+}$ for which $P_{lj} = \mathrm{tr}_j$ \textbf{:} 
    \State \quad $z_j[k+1]$ and $y_j[k+1]$ then,
    \State \quad $y_j[k + 1] = 0$, $z_j[k + 1] = 0$
    \State \textbf{else}
    \State \quad $y_j^s[k+1] = y_j^s[k], z_j^s[k+1] = z_j^s[k]$
    \State \textbf{end}
    \State \textbf{Output:} $q_j^s[k + 1] = \frac{y_j^s[k + 1]}{z_j^s[k + 1]}$
\EndFor
\end{algorithmic}
\end{varalgorithm}

\begin{remark}
The event-trigger conditions \textbf{C1} and \textbf{C2} in Algorithm~\ref{alg:arq_quantized} enable finite time convergence to the exact quantized average of the initial states (as also shown in \cite{2020:Rikos_Quant_Cons}). 
However, the main difference of Algorithm~\ref{alg:arq_quantized} compared to \cite[Algorithm~$1$]{2020:Rikos_Quant_Cons} is its ability to handle packet drops.
\end{remark}

\begin{example}\label{example:1}
We provide an example to demonstrate the operation of Algorithm 1.
Consider the strongly connected digraph $\mathcal{G}_d = (\mathcal{V}, \mathcal{E})$ in Fig.~\ref{fig:enter-label}, where $\mathcal{V} = \{v_1, v_2, v_3, v_4\}$ and $\mathcal{E} = \{ \varepsilon_{12}, \varepsilon_{23}, \varepsilon_{31}, \varepsilon_{32}, \varepsilon_{34}, \varepsilon_{41}\}$. Each node is initialized with an integer value such that $y_1[0] = 2$, $y_2[0] = 6$, $y_3[0] = 2$, and $y_4[0] = 6$, respectively. Hence, the average consensus value is equal to $\hat{q} = 4$.

\begin{figure}[h]
    \centering
    \includegraphics[width=0.8\linewidth]{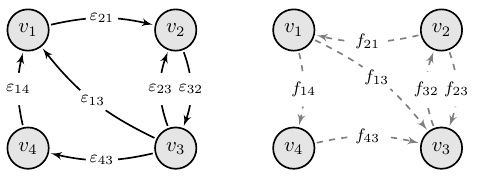}
    \caption{The digraph considered in Example 1 (left), along with its associated feedback links (right).}
    \label{fig:enter-label}
\end{figure}

In this example, we assume that each node $v_j \in \mathcal{V}$ incorporates an HARQ mechanism with maximum retransmission limit {$\bar{\tau}_{ij}=\bar{\tau}=2$} for all $v_j \in \mathcal{N}^{+}_{i}$. Moreover, as already discussed in previous sections, the packet error probability reduces exponentially with the number of retransmissions as stated in \eqref{eq:harq_packet_error_probability} where we assume that for a packet initially transmitted over the link $\varepsilon_{ji}$ at time $k$, the initial packet error probability is $p_{ji}=p=0.6$ for all $\varepsilon_{ji} \in \mathcal{E}$. 

Now, let us delve into the operation of the HARQ-QAC algorithm based on the setting we just introduced, with the aid of Tables~\ref{tab:k_0}-\ref{tab:k_8}\footnote{Note that we do not require agents to maintain buffers for links that are not available in the digraph of Example~\ref{example:1}. For this reason, we denote by $\bullet$ the absence of such buffers.}.
At time step $k=0$ each node $v_j$ transmits its initial mass variables $y_j[0]$ and $z_j[0]$ according to digraph $\mathcal{G}_d$ and the event-triggered conditions of the Algorithm~\ref{alg:arq_quantized}. By the end of time step $k=0$, the nodes receive back the corresponding HARQ feedback of the packets they transmitted. During this iteration, only the packet sent from node $v_2$ to node $v_1$ has been successfully received. Thus, nodes $v_1$, $v_3$, and $v_4$ store the packets that arrive in error in their corresponding buffers to be retransmitted in the next iteration, as shown in Table~\ref{tab:k_1}.


\begin{table}[h!]
\centering
\begin{footnotesize}
\begin{tabular}{|C{1.33cm}" C{0.5cm}|C{0.5cm}|C{0.5cm}|C{0.5cm}|C{1.2cm}|C{1.2cm}|}
\hline
\textbf{Node} & \( y_j[0] \) & \( z_j[0] \) & \( y_j^{s}[0] \) & \( z_j^{s}[0] \) & \( y_{ji}^{b}[0] \) & \(  z_{ji}^{b}[0]\) \\ \thickhline
\rowcolor[HTML]{cccccc} $v_1$ & $2$ & $1$ & $2$ & $1$ & $[\bullet~0~\bullet~\bullet]$ & $[\bullet~0~\bullet~\bullet]$  \\ 
\rowcolor[HTML]{eeeeee} $v_2$ & $6$ & $1$ & $6$ & $1$ & $[\bullet~\bullet~0~\bullet]$ & $[\bullet~\bullet~0~\bullet]$  \\ 
\rowcolor[HTML]{cccccc} $v_3$ & $2$ & $1$ & $2$ & $1$ & $[0~0~\bullet~0]$ & $[0~0~\bullet~0]$  \\ 
\rowcolor[HTML]{eeeeee} $v_4$ & $6$ & $1$ & $6$ & $1$ & $[0~\bullet~\bullet~\bullet]$ & $[0~\bullet~\bullet~\bullet]$  \\ \hline
\end{tabular}
\end{footnotesize}
\vspace{6pt}
\caption{Operation of the HARQ-QAC algorithm at $k=0$.}
\label{tab:k_0}
\vspace{-10pt}
\end{table}

By time step $k=1$, node $v_3$ receives the mass variables from $v_1$, and hence it will transmit its new updated mass to its next out-neighbor $v_2$ (according to the predetermined order), as well as the mass stored in its buffers (due to the packet error at $k=0$) to node $v_1$. Nodes $v_1$, $v_2$, and $v_4$, however, do not transmit any new packet since they did not receive any information in the previous iteration. Yet, they retransmit their corresponding packets held in their buffers (if any) to their destined out-neighbors. The updated states of the nodes based on the transmissions and retransmissions during iteration $k=1$ are shown in Table~\ref{tab:k_2}.



\begin{table}[h!]
\centering
\begin{footnotesize}
\begin{tabular}{|C{1.33cm}" C{0.5cm}|C{0.5cm}|C{0.5cm}|C{0.5cm}|C{1.2cm}|C{1.2cm}|}
\hline
\textbf{Node} & \( y_j[1] \) & \( z_j[1] \) & \( y_j^{s}[1] \) & \( z_j^{s}[1] \) & \( y_{ji}^{b}[1] \) & \(  z_{ji}^{b}[1]\) \\ \thickhline
\rowcolor[HTML]{cccccc} $v_1$ & $0$ & $0$ & $2$ & $1$ & $[\bullet~2~\bullet~\bullet]$ & $[\bullet~1~\bullet~\bullet]$  \\ 
\rowcolor[HTML]{eeeeee} $v_2$ & $0$ & $0$ & $6$ & $1$ & $[\bullet~\bullet~0~\bullet]$ & $[\bullet~\bullet~0~\bullet]$  \\ 
\rowcolor[HTML]{cccccc} $v_3$ & $6$ & $1$ & $6$ & $1$ & $[2~0~\bullet~0]$ & $[1~0~\bullet~0]$  \\ 
\rowcolor[HTML]{eeeeee} $v_4$ & $0$ & $0$ & $6$ & $1$ & $[6~\bullet~\bullet~\bullet]$ & $[1~\bullet~\bullet~\bullet]$  \\ \hline
\end{tabular}
\end{footnotesize}
\vspace{6pt}
\caption{Operation of the HARQ-QAC algorithm at $k=1$.}
\label{tab:k_1}
\vspace{-10pt}
\end{table}

\begin{table}[h!]
\centering
\begin{footnotesize}
\begin{tabular}{|C{1.33cm}" C{0.5cm}|C{0.5cm}|C{0.5cm}|C{0.5cm}|C{1.2cm}|C{1.2cm}|}
\hline
\textbf{Node} & \( y_j[2] \) & \( z_j[2] \) & \( y_j^{s}[2] \) & \( z_j^{s}[2] \) & \( y_{ji}^{b}[2] \) & \(  z_{ji}^{b}[2]\) \\ \thickhline
\rowcolor[HTML]{cccccc} $v_1$ & $8$ & $2$ & $8$ & $2$ & $[\bullet~2~\bullet~\bullet]$ & $[\bullet~1~\bullet~\bullet]$  \\ 
\rowcolor[HTML]{eeeeee} $v_2$ & $0$ & $0$ & $6$ & $1$ & $[\bullet~\bullet~0~\bullet]$ & $[\bullet~\bullet~0~\bullet]$  \\ 
\rowcolor[HTML]{cccccc} $v_3$ & $0$ & $0$ & $6$ & $1$ & $[0~6~\bullet~0]$ & $[0~1~\bullet~0]$  \\ 
\rowcolor[HTML]{eeeeee} $v_4$ & $0$ & $0$ & $6$ & $1$ & $[0~\bullet~\bullet~\bullet]$ & $[0~\bullet~\bullet~\bullet]$  \\ \hline
\end{tabular}
\end{footnotesize}
\vspace{6pt}
\caption{Operation of the HARQ-QAC algorithm at $k=2$.}
\label{tab:k_2}
\vspace{-10pt}
\end{table}

A similar procedure is followed for the next iterations, which we omit for ease of presentation. However, to emphasize the convergence of the HARQ-QAC algorithm, we present the last two iterations (in Tables~\ref{tab:k_7}-\ref{tab:k_8} by which the nodes converge to the exact average consensus value. In particular, at iteration $k=7$ all nodes $v_j \in \mathcal{V}$ converged already to the average consensus value since $q^s_j[7]=y^s_j[7]/z^s_j[7]=4$. However, due to the triggering conditions of the HARQ-QAC, node $v_4$ should transmit its mass variables to node $v_1$ which are successfully received at time $k=8$. Consequently, since no further triggering conditions hold for the next iterations and the buffers of the individual nodes are empty, the algorithm terminates.

\begin{table}[h!]
\centering
\begin{footnotesize}
\begin{tabular}{|C{1.33cm}" C{0.5cm}|C{0.5cm}|C{0.5cm}|C{0.5cm}|C{1.2cm}|C{1.2cm}|}
\hline
\textbf{Node} & \( y_j[7] \) & \( z_j[7] \) & \( y_j^{s}[7] \) & \( z_j^{s}[7] \) & \( y_{ji}^{b}[7] \) & \(  z_{ji}^{b}[7]\) \\ \thickhline
\rowcolor[HTML]{cccccc} $v_1$ & $0$ & $0$ & $8$ & $2$ & $[\bullet~0~\bullet~\bullet]$ & $[\bullet~0~\bullet~\bullet]$  \\ 
\rowcolor[HTML]{eeeeee} $v_2$ & $0$ & $0$ & $16$ & $4$ & $[\bullet~\bullet~0~\bullet]$ & $[\bullet~\bullet~0~\bullet]$  \\ 
\rowcolor[HTML]{cccccc} $v_3$ & $0$ & $0$ & $16$ & $4$ & $[0~0~\bullet~0]$ & $[0~0~\bullet~0]$  \\ 
\rowcolor[HTML]{eeeeee} $v_4$ & $16$ & $4$ & $16$ & $4$ & $[0~\bullet~\bullet~\bullet]$ & $[0~\bullet~\bullet~\bullet]$  \\ \hline
\end{tabular}
\end{footnotesize}
\vspace{0pt}
\caption{Operation of the HARQ-QAC algorithm at $k=7$.}
\label{tab:k_7}
\vspace{-5pt}
\end{table}

\begin{table}[h!]
\centering
\begin{footnotesize}
\begin{tabular}{|C{1.33cm}" C{0.5cm}|C{0.5cm}|C{0.5cm}|C{0.5cm}|C{1.2cm}|C{1.2cm}|}
\hline
\textbf{Node} & \( y_j[8] \) & \( z_j[8] \) & \( y_j^{s}[8] \) & \( z_j^{s}[8] \) & \( y_{ji}^{b}[8] \) & \(  z_{ji}^{b}[8]\) \\ \thickhline
\rowcolor[HTML]{cccccc} $v_1$ & $0$ & $0$ & $16$ & $4$ & $[\bullet~0~\bullet~\bullet]$ & $[\bullet~0~\bullet~\bullet]$  \\ 
\rowcolor[HTML]{eeeeee} $v_2$ & $0$ & $0$ & $16$ & $4$ & $[\bullet~\bullet~0~\bullet]$ & $[\bullet~\bullet~0~\bullet]$  \\ 
\rowcolor[HTML]{cccccc} $v_3$ & $0$ & $0$ & $16$ & $4$ & $[0~0~\bullet~0]$ & $[0~0~\bullet~0]$  \\ 
\rowcolor[HTML]{eeeeee} $v_4$ & $0$ & $0$ & $16$ & $4$ & $[0~\bullet~\bullet~\bullet]$ & $[0~\bullet~\bullet~\bullet]$  \\ \hline
\end{tabular}
\end{footnotesize}
\vspace{0pt}
\caption{Operation of the HARQ-QAC algorithm at $k=8$.}
\label{tab:k_8}
\vspace{-10pt}
\end{table}

\end{example}

\begin{prop}\label{proof_theorem}
Consider a strongly connected digraph $\mathcal{G}_d =(\mathcal{V}, \mathcal{E})$ with $n=|\mathcal{V}|$ nodes and $m=|\mathcal{E}|$ edges. Suppose that each node $v_j\in\mathcal{V}$ executes Algorithm~\ref{alg:arq_quantized}. Then, we can find a finite number of steps $k_1\in\mathbb{N}$, such that for $k\geq k_1$ each node $v_j\in\mathcal{V}$ reaches
\begin{align*}
    y^s_j[k] = \sum_{\ell\in\mathcal{V}}y_{\ell}[0] \quad \text{and} \quad z^s_j[k]=n,
\end{align*}
almost surely (\emph{i.e.,} with probability 1). 
\end{prop}
\begin{proof}
The proof follows a similar structure \emph{mutatis mutandis} to that of~\cite[Proposition~2]{2018:RikosHadj_CDC}, with the main difference that in this proof we also incorporate the effect of the packet losses and the use of the HARQ mechanism.
\end{proof}

%
%
%
%
\section{Simulation Results}\label{sec:results}

To validate the performance of the proposed algorithm, we run numerical simulations based on the setting considered in Example~\ref{example:1}. In what follows, however, we further consider the following three scenarios. 
\begin{list4}
\item \textbf{SC1:} Ideal channel conditions in which no packet errors are assumed, \emph{i.e.,} $p_{ji}=0$ for all $\varepsilon_{ji}\in\mathcal{E}$, which reduces to the algorithm in \cite{2018:RikosHadj_CDC}. 
\item \textbf{SC2:} Unreliable channels with packet error probability $p_{ji}=0.6$ for all $\varepsilon_{ji}\in\mathcal{E}$, using HARQ-QAC with $\lambda=1$ and $\bar{\tau}=2$, which reduces to an ARQ variant of the proposed algorithm.
\item \textbf{SC3:} Unreliable channels with initial packet error probability $p_{ji}=0.6$ for all $\varepsilon_{ji}\in\mathcal{E}$, using HARQ-QAC with $\lambda=0.3$ and $\bar{\tau}=2$.    
\end{list4}

\begin{figure}[h]
    \centering
    \includegraphics[width=\linewidth]{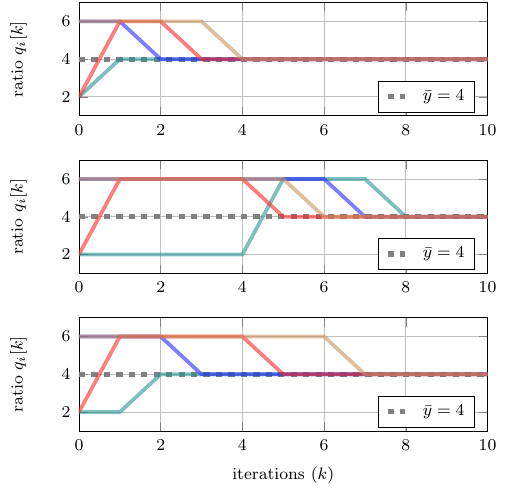}
    \caption{Ratio of $v_i\in\mathcal{V}$. Upper: SC1, middle: SC2, and lower: SC3.}
    \label{fig:results_ratios}
\end{figure}

In Fig.~\ref{fig:results_ratios}, we illustrate the ratio for all $v_i\in\mathcal{V}$, over 10 iterations under the three different scenarios, \emph{i.e.,}  SC1 (upper plot), SC2 (middle plot), and SC3 (lower plot). 
Clearly, in all the scenarios, agents converge to the exact average consensus value, albeit the packet errors in SC2 and SC3. However, as also shown in Fig.~\ref{fig:results_errors}, in HARQ-QAC with $\lambda=0.3$, \emph{i.e.,} SC3, achieves faster convergence to the average of the nodes initial values compared to SC2 where $\lambda=1$. This behavior is expected since receiving nodes exploit information from previous retransmission trials to reduce the probability of decoding errors at each iteration. Notice also that, the average consensus error for each scenario decays exactly to $0$ due to the quantized values exchanged over the network (assuming that they are initialized at integer values).

In Fig.~\ref{fig:results_errors}, we demonstrate the average consensus error over time, for the three scenarios (SC1, SC2, and SC3). SC1, denoted by black, shows a quick decrease in error, which suggests that this scenario quickly moves toward consensus. After about 4 iterations, the error drops to 0. For SC2, indicated by red, the error reduction is more gradual, it reaches stability in error after the 8th iteration. Lastly, in SC3, marked by cyan, the error decreases and achieves stability in error after the 7th iteration, reflecting a faster approach to consensus compared to SC2.

\begin{figure}[h]
    \centering
    \includegraphics{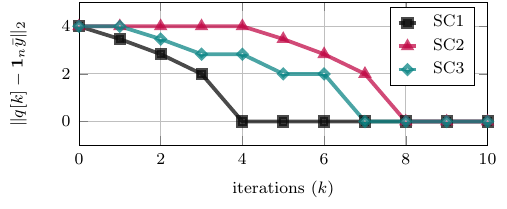}
    \caption{Average consensus error for the three considered scenarios.}
    \label{fig:results_errors}
\end{figure}

%
%
%
%
Here, we evaluate the number of iterations required for convergence in a 20 node network across various error rates and $\lambda$ values with $\bar{\tau}=3$, in terms of the number of iterations. 
As shown in TABLE~\ref{tab:k_8}, HARQ-QAC with $\lambda<1$ outperforms the HARQ-QAC with $\lambda=1$, especially at a higher error rate.

\begin{table}[h!]
\centering
\begin{footnotesize}
\begin{tabular}
{|C{1.33cm}" C{.87cm}|C{.965cm}|C{.965cm}|C{.965cm}|C{.87cm}|}
\hline
\textbf{Error Rate} & \multicolumn{5}{c|}{\bf Iterations} \\
 (\%) & $\lambda=0$ & $\lambda=0.1$ & $\lambda=0.5$ & $\lambda=0.8$ & $\lambda=1$ \\ \thickhline
\rowcolor[HTML]{cccccc} 0  & 92  & 92  & 92  & 92  & 92  \\ 
\rowcolor[HTML]{eeeeee} 20 & 106 & 159 & 159 & 170 & 179 \\ 
\rowcolor[HTML]{cccccc}40 & 138 & 138 & 208 & 185 & 189 \\ 
\rowcolor[HTML]{eeeeee}60 & 144 & 153 & 150 & 204 & 213 \\ 
\rowcolor[HTML]{cccccc}80 & 164 & 171 & 195 & 299 & 516 \\ \hline
\end{tabular}
\end{footnotesize}
\vspace{6pt}
\caption{Convergence with varying error rates and $\lambda$ values.}
\label{tab:k_8}
\vspace{-10pt}
\end{table}

Figure 4 shows the relationship between different $\lambda$ values and the number of iterations required for convergence at 60\% and 80\% error rate. As $\lambda$ increases, the number of iterations rises significantly, highlighting the impact of $\lambda$ on the convergence time of the algorithm.


\begin{figure}[h!]
\centering
    \includegraphics{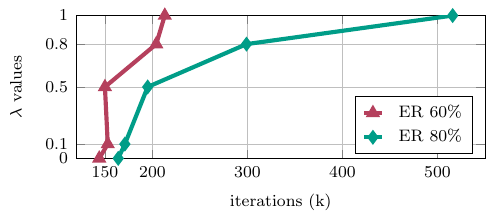}
    \vspace{-5pt}
    \caption{Convergence for error rates 60\% and 80\%}
\end{figure}

%
%
%
%
\section{Conclusions and Future Directions}\label{sec:conclusions}


In this paper, we considered the problem of average consensus in a digraph when the nodes exchange quantized packets in the presence of packet losses. We proposed a distributed algorithm based on that of \cite{2018:RikosHadj_CDC} and incorporated a HARQ mechanism with which transmitting nodes can a) receive ACKs/NACKs about the reception of the packet transmitted, and b) decrease the probability of packet loss when a packet is retransmitted, thus accelerating the convergence to the average. The way the algorithm works was demonstrated via a simple example and simulations. It was shown that our proposed algorithm converges to the average of the initial quantized values in a finite number of steps, despite the packet losses. 
handle packet losses and converge to the average.


This work brings forward several challenges for deploying communication protocols in distributed settings. As an example, it is of high interest to investigate how packet losses and the allowable number of retransmissions affect the convergence of the algorithm, under diverse network conditions (e.g., different sizes, different diameters, etc).
\bibliographystyle{IEEEtran}
\bibliography{references}

%
%
%
%
\end{document}